\documentclass[12pt]{article}%
\usepackage{graphicx}%
\usepackage[normalem]{ulem}
\usepackage{amsmath,amsthm,amsfonts,
amsbsy,amssymb,upref,enumerate,bigstrut,
color,mathtools,mathrsfs,float,bm,dsfont,scalefnt,mathtools}
\usepackage{natbib}
\usepackage{lipsum}
\usepackage{stmaryrd}
\usepackage{authblk} 
\usepackage{subfigure}
\usepackage{pgfplots}
\usetikzlibrary{pgfplots.fillbetween}

\usepackage{newtxtext,newtxmath} 
\usepackage{mhchem} 

\usepackage{caption}
\usepackage{tikz}
\usetikzlibrary{intersections}

\usepackage[colorlinks=true,linkcolor=blue,citecolor=blue]{hyperref}

\usepackage[left=1in,top=1.2in,right=0.7in]{geometry}

\newtheorem{theorem}{Theorem}[section]

\newtheorem{corollary}[theorem]{Corollary}

\newtheorem{definition}{Definition}[section]
\newtheorem{lemma}[theorem]{Lemma}
\newtheorem{proposition}[theorem]{Proposition}
\newtheorem{remark}[theorem]{Remark}
\numberwithin{equation}{section} 

\makeatletter
\def\@seccntformat#1{\@ifundefined{#1@cntformat}%
	{\csname the#1\endcsname\quad}
	{\csname #1@cntformat\endcsname}
}
\makeatother

\newif\ifShowComments
\ShowCommentstrue
\def\strutdepth{\dp\strutbox}
\def\druk#1{\strut\vadjust{\kern-\strutdepth
        {\vtop to \strutdepth{%
                \baselineskip\strutdepth\vss
                        \llap{\hbox{#1}\quad}\null}}}}

\title{\bf
Income inequality estimation with gamma mixtures
}

\author[1]{Roberto Vila\thanks{rovig161@gmail.com \, (Corresponding author)}}
\author[1,2]{Helton Saulo\thanks{heltonsaulo@gmail.com \,}} 
\author[1]{Felipe Quintino\thanks{felipesquintino2@gmail.com}}
\affil[1]{Department of Statistics, University of
	 Bras\'ilia, Bras\'ilia, Brazil}
\affil[2]{
Department of Economics, Federal University of Pelotas, Pelotas, Brazil}

\begin{document}
\maketitle

\begin{abstract}
This paper studies the estimation of the $m$th Gini index under finite mixtures of gamma distributions. We derive closed-form expressions for the $m$th Gini index and for the expectation and bias of its non-parametric U-statistic estimator, extending previous results for both single gamma populations \citep{Baydil2025,VilaSaulomth2025} and gamma mixture models \citep{VilaSaulomixture2025}. We further establish the asymptotic properties of the estimator for gamma mixtures sharing a common rate parameter, including an asymptotic lower bound for the bias, asymptotic unbiasedness, strong consistency, and asymptotic normality. Although these theoretical results require a common rate parameter, a Monte Carlo study also investigates the estimator under mixtures with different rates and compares its performance with bias-corrected and parametric estimators. Finally, the proposed methodology is illustrated through the analysis of an income dataset.
\end{abstract}


\section{Introduction}

Income data are well known to exhibit right-skewed distributions with a long tail corresponding to a relatively small proportion of high-income observations. To accommodate this behavior, finite mixture models have become a popular alternative to single distributions, allowing heterogeneous populations to be represented through a collection of homogeneous subpopulations. Mixture models have been successfully applied to income data using Lognormal mixtures \cite[cf.][]{flachaire2007estimation, lubrano2016income}, beta-2 mixtures \cite[cf.][]{chotikapanich2007estimating}, gamma mixtures (cf. \cite{Chotikapanich2008}), and Weibull mixtures \cite[cf.][]{bakar2022income}. A comprehensive overview of these models and their applications to income distributions is provided by \cite{chotikapanich2008modeling}.

Among the many measures of economic inequality, the Gini coefficient remains one of the most widely used due to its simple interpretation and attractive theoretical properties. Recently, \cite{Gavilan-Ruiz2024} introduced the $m$th Gini index, which generalizes the classical Gini coefficient by replacing pairwise comparisons with the range of samples of size $m$. This new index preserves several important properties of the Gini coefficient while providing a broader family of inequality measures.

From a statistical viewpoint, finite-sample bias is an important issue in the estimation of inequality measures. For the classical Gini coefficient,  \cite{Deltas2003} proposed a U-statistic estimator, and  \cite{Baydil2025} showed that it is exactly unbiased under gamma populations. More recently, \cite{VilaSaulomth2025}  established the corresponding unbiasedness result for the $m$th Gini index under the gamma distribution, while \cite{VilaSaulomixture2025} derived analytical expressions for the bias of the Gini coefficient under gamma mixture models.

The main objective of this paper is to investigate the estimation of the $m$th Gini index under finite mixtures of gamma distributions. We first derive a closed-form expression for the $m$th Gini index and then obtain explicit formulas for the expectation and bias of its non-parametric U-statistic estimator. For mixtures sharing a common rate parameter, we establish an asymptotic lower bound for the bias, prove asymptotic unbiasedness, strong consistency, and asymptotic normality of the estimator. Although these theoretical results rely on the common-rate assumption, our Monte Carlo study also examines mixtures with different rate parameters and compares the proposed estimator with bias-corrected and parametric alternatives. Finally, we illustrate the methodology with an application to an income dataset.

The remainder of the paper is organized as follows. Section~\ref{sec:02} reviews the $m$th Gini index and derives its expression under gamma mixture distributions. Section~\ref{Deriving estimator biases} establishes explicit formulas for the expectation and bias of the proposed estimator. Section~\ref{Asymptotic lower bound for the bias} derives an asymptotic lower bound for the bias, while Sections~\ref{Asymptotic unbiasedness} and \ref{Strong consistency and asymptotic normality} establish its asymptotic unbiasedness, strong consistency, and asymptotic normality. Section~\ref{Illustrative simulation study} presents a Monte Carlo study comparing the proposed estimator with bias-corrected and parametric competitors under several gamma mixture scenarios, including mixtures with different rate parameters. Section~\ref{sec:applications} illustrates the methodology through the analysis of an income dataset. Finally, Section~\ref{concluding_remarks} concludes the paper.

\section{Preliminaries and definitions}\label{sec:02}


Let $X_1, \ldots, X_m$ be independent and identically distributed (iid) random variables with the same distribution as a non-negative random variable $X$ such that $\mu = \mathbb{E}[X] > 0$. Following Definition~1 in \cite{Gavilan-Ruiz2024}, for each integer $m \geqslant 2$, the $m$th Gini index of $X$ is defined as
\begin{align}\label{m-pop-Gini}
	IG_m \equiv IG_m(X)
	= \frac{\mathbb{E}\left[X_{(m)} - X_{(1)}\right]}{m\mu},
\end{align}
where $X_{(1)}=\min\{X_1,\ldots,X_m\}$ and $X_{(m)}=\max\{X_1,\ldots,X_m\}$ are the extremal order statistics.

Observe that for $m=2$, the $m$th Gini index coincides with the traditional Gini coefficient \citep{Gini1936}, denoted $G$, namely,
\begin{align}\label{Gini coefficient}
	G \equiv G(X) \equiv IG_2={\mathbb{E}\vert X_1-X_2\vert\over 2\mu}.
\end{align}

Below, we present some results, which formalize the essential properties and characterizations of the $m$th Gini index \citep[see][]{Gavilan-Ruiz2024}:
\begin{itemize}
	\item[(P1)]	The expectations $\mathbb{E}[X_{(1)}]$ and $\mathbb{E}[X_{(m)}]$ exist, and therefore the $m$th Gini index $IG_m$  exists for any integer $m\geqslant 2$.
	\item[(P2)] $0\leqslant IG_m<1$, for all $m \geqslant 2$.
	\item[(P3)] The $m$th Gini index is invariant under positive scalar multiplication of the variable $X$; that is, for any $b > 0$,	
	$
	IG_m(bX)=IG_m(X).
	$
	\item[(P4)] 	The $m$th Gini index is affected by additive shifts in \( X \), meaning it is not translation invariant. More precisely, for any \( a > 0 \),
	$
	IG_m(a + X) = \left[{\mu}/{(a + \mu)}\right] IG_m(X),
	$
	where \( \mu \) denotes the mean of \( X \).
	\item[(P5)] The $m$th Gini index can be interpreted as the covariance between $X$ and a transformation of $X$; that is,
	\begin{align*}
		IG_m
		=
		{1\over\mu} \,
		{\rm Cov}\left(X,F_{X_{(m-1)}}(X)+F_{X_{(1)}}(X)\right),
        \quad m\geqslant 2,
	\end{align*}
	where \( \mu \) is the mean of \( X \), $F_{X_{(m-1)}}(x)=F^{m-1}(x)$ and $F_{X_{(1)}}(x)=1-[1-F(x)]^{m-1}$ are the cumulative distribution functions (CDFs)  of the order statistics $X_{(m-1)}=\max\{X_1,\ldots,X_{m-1}\}$ and $X_{(1)}=\max\{X_1,\ldots,X_{m-1}\}$, respectively.
	By taking $m = 2$ in the preceding formula, we obtain the classical identity for the Gini coefficient \eqref{Gini coefficient} \citep[see][]{Yin2024}:
	$
	G = {R_G(F)}/{\mu},
	$
	where
	$
	R_G(F) \equiv ({1}/{2})\,\mathbb{E}|X_1 - X_2|
	= \operatorname{Cov}\!\left(X,\, 2F(X)-1\right)
	= \int_0^{1} F^{-1}(p)\,(2p - 1)\, \mathrm{d}p
	$
	is the Gini mean difference (GMD).
	
	\item[(P6)] The Lorenz curve $L(\cdot)$ for $X$ is defined by 
	$
		L(p)={\int_0^p F^{-1}(t){\rm d}t/\mu}, \ \text{for any} \ 0\leqslant p\leqslant 1,
	$
	with $\mu$ being the mean of \( X \).
	The $m$th Gini index can be written in terms of Lorenz measures of inequality as follows:
	\begin{align*}
		IG_m
		=
		\left(1-{1\over m}\right) 
		D_{m-1}(F)
		+
		{1\over m}\, G_m(F),
	\end{align*}
	where
	$
		D_n\equiv D_n(F)=(n+1)\mathbb{E}[\{U-L(U)\}U^{n-1}], 
		\ U\sim \mathrm{U}(0,1), \ n\geqslant 1,
	$
	is the Lorenz measure of inequality (for the CDF $F$ of $X$) introduced in \cite{Aaberge2000}, and
	$
		G_n
		\equiv
		G_n(F)
		=
		n(n-1)\mathbb{E}[\{U-L(U)\}(1-U)^{n-2}], 
		\ U\sim \mathrm{U}(0,1), \ n\geqslant 1,
	$
	is the generalized Gini measure introduced in \cite{Donaldson1980}, \cite{Kakwani1980} and \cite{Yitzhaki1983}.
	\item[(P7)]	
	Suppose $X$ is a non-negative random variable with mean $\mu>0$. Then, for every integer $m\geqslant 2$ one can find a value $r_m> 0$
	satisfying 
	$
		IG_m(X)=G(r_m+X),
	$
	where $G$ is the Gini coefficient \eqref{Gini coefficient}.
	Thus, the $m$th Gini index $IG_m$ may be regarded as a genuine Gini coefficient, interpretable through the conventional Gini construct. Specifically, $IG_m$ for a random variable $X$ equals the standard Gini coefficient $G$ evaluated at $r_m+X$, where the quantity $r_m$ determines the adjustment magnitude.
\end{itemize}

\begin{definition}[Mixture of gamma distributions] \label{definition:mixture-gamas}
	A random variable $X$ is said to follow a gamma mixture distribution with parameter vector
	$
	\bm{\theta} = (\pi_1, \dots, \pi_{k-1}, \alpha_1, \dots, \alpha_k, \lambda)^\top,
	$
	denoted by $X \sim \mathrm{GM}(\bm{\theta})$, if its probability density function is given by
	\begin{equation}
		\label{eq:mixture-gama}
		f_X(x; \bm{\theta}) = \sum_{j=1}^{k} \pi_j f_{Z_j}(x; \alpha_j, \lambda), \quad x > 0,
	\end{equation}
	where $k \in \mathbb{N}$ is the number of mixture components and $\pi_j$ are the mixing proportions
	satisfying $\pi_j > 0$ and $\pi_k = 1 - \sum_{j=1}^{k-1} \pi_j$.
	Moreover, here, $f_{Z_j}(x; \alpha_j, \lambda)$ denotes the density of the gamma-distributed random variable
	$Z_j \sim \mathrm{Gamma}(\alpha_j, \lambda)$, with shape parameter $\alpha_j > 0$ and rate parameter
	$\lambda > 0$, that is,
	\begin{equation*}
		f_{Z_j}(x; \alpha_j, \lambda)
		=
		\frac{\lambda^{\alpha_j}}{\Gamma(\alpha_j)} \,
		x^{\alpha_j - 1} \exp\{-\lambda x\}, \quad x > 0,
	\end{equation*}
	where $\Gamma(\cdot)$ denotes the complete gamma function.
\end{definition}

\begin{remark}
    The $GM(\boldsymbol{\theta})$ family admits a reparameterization based on the scale parameter $\beta$, exploiting the identity $\beta = 1/\lambda$.
\end{remark}

Figure 1 in \cite{VilaSaulomixture2025} demonstrates how the gamma mixture distribution constitutes a flexible framework that subsumes numerous key special cases.

{\color{black}
	It is straightforward to see that the random variable $X \sim \mathrm{GM}(\bm{\theta})$ admits the representation
	\begin{align}\label{representation-X}
		X = \sum_{j=1}^{m} \mathds{1}_{\{Y=j\}} Z_j,
	\end{align}
	where $Y \in \{1, \dots, m\}$ is a discrete random variable with $\mathbb{P}(Y=j) = \pi_j$, independent of $Z_j \sim \mathrm{Gamma}(\alpha_j, \lambda)$ for each $j = 1, \dots, m$, and $\mathds{1}_A$ denotes the indicator function of the event $A$.
	Then, the CDF of $X \sim \text{GM}(\bm{\theta})$ is given by
	\begin{align}
		\label{cdf-eq:mixture-gama}
		F_X(x; \bm{\theta}) 
		= 
		\sum_{j=1}^{k} \pi_j F_{Z_j}(x; \alpha_j, \lambda), \quad x > 0.
	\end{align}

	\begin{proposition}
		\label{proposition:ev-gx-11}
		If $X \sim \mathrm{GM}(\bm{\theta})$, then, for any Borel-measurable function $h$ we have 
		\[ 
		\mathbb{E}[h(X)] 
		= 
		\sum_{j=1}^k\pi_j\mathbb{E}[h(Z_j)],
		\]
		where
		$Z_j\sim{\rm Gamma}(\alpha_j,\lambda)$ for $j=1,\ldots,k$, provided the above expectations exist.
	\end{proposition}
	\begin{proof}
		The proof follows by combining \eqref{cdf-eq:mixture-gama} with the definition of expectation $\mathbb{E}[h(X)] =\int_0^\infty h(x){\rm d}F_X(x; \bm{\theta})$.
	\end{proof}

	We are now ready to prove the principal result of this section.
	\begin{theorem}\label{main-theo-0}
		The $IG_m$ index of $X \sim \mathrm{GM}(\bm{\theta})$ can be characterized as follows:
		\begin{multline*}
			IG_m = IG_m(\boldsymbol{\theta})=
			\left[{1\over m \sum_{j=1}^{k} \pi_j \alpha_j}\right]
			\sum_{\substack{r_1 + \cdots + r_k = m \\
					r_1,\ldots,r_k\geqslant 0}}
			\binom{m}{r_1, \ldots, r_k}
			\left(
			\prod_{j=1}^{k}	
			\pi_j^{r_j} 
			\right)
			\\[0,2cm]
			\times 
			\int_0^\infty 
			\left[
			1-
			\prod_{j=1}^{k}
			{\gamma^{r_j}(\alpha_j,u)\over \Gamma^{r_j}(\alpha_j)}
			-
			\prod_{j=1}^{k}
			{\Gamma^{r_j}(\alpha_j,u)\over \Gamma^{r_j}(\alpha_j)}
			\right]
			{\rm d}u,
		\end{multline*}
		where $\Gamma(a)$, $\gamma(a,b)$, and $\Gamma(a,b)$ represent the complete, lower incomplete, and upper incomplete gamma functions, respectively, 	and
		\begin{align*}
			{\displaystyle {m \choose r_{1},\ldots ,r_{k}}={\frac {m!}{r_{1}!\cdots r_{k}!}}}   
		\end{align*}
		is a multinomial coefficient. 
		The sum is over all nonnegative integers $r_1,\ldots,r_k$ with total sum $m$.
	\end{theorem}
    \begin{proof}
Since
$
\mathbb{P}(X_{(m)}\leqslant t)=F_X^m(t;\bm\theta)
$
and
$
\mathbb{P}(X_{(1)}>t)=\{1-F_X(t;\bm\theta)\}^m,
$
it follows that
\[
IG_m=
\frac{1}{m\mu}
\int_0^\infty
\left[
1-F_X^m(t;\bm\theta)
-
\{1-F_X(t;\bm\theta)\}^m
\right]{\rm d}t.
\]

Now, using the mixture representation \eqref{cdf-eq:mixture-gama}
together with the multinomial theorem,
\[
\left(\sum_{j=1}^k a_j\right)^m
=
\sum_{\substack{r_1+\cdots+r_k=m\\r_j\geqslant 0}}
\binom{m}{r_1,\ldots,r_k}
\prod_{j=1}^k a_j^{\,r_j},
\]
we obtain
\[
1-F_X^m(t;\bm\theta)
=
\sum_{\substack{r_1+\cdots+r_k=m\\r_j\geqslant 0}}
\binom{m}{r_1,\ldots,r_k}
\prod_{j=1}^k\pi_j^{r_j}
\left[
1-
\prod_{j=1}^kF_{Z_j}^{\,r_j}(t;\alpha_j,\lambda)
\right],
\]
where we used
\[
\sum_{\substack{r_1+\cdots+r_k=m\\r_j\geqslant 0}}
\binom{m}{r_1,\ldots,r_k}
\prod_{j=1}^k\pi_j^{r_j}
=
\left(\sum_{j=1}^k\pi_j\right)^m
=1.
\]
Similarly,
\[
\{1-F_X(t;\bm\theta)\}^m
=
\sum_{\substack{r_1+\cdots+r_k=m\\r_j\geqslant 0}}
\binom{m}{r_1,\ldots,r_k}
\prod_{j=1}^k\pi_j^{r_j}
\prod_{j=1}^k
\{1-F_{Z_j}(t;\alpha_j,\lambda)\}^{r_j}.
\]

Substituting these expansions into the integral yields
\[
IG_m
=
\frac{\displaystyle
\sum_{\substack{r_1+\cdots+r_k=m\\r_j\geqslant 0}}
\binom{m}{r_1,\ldots,r_k}
\prod_{j=1}^k\pi_j^{r_j}
\int_0^\infty
\left[
1-
\prod_{j=1}^kF_{Z_j}^{\,r_j}(t;\alpha_j,\lambda)
-
\prod_{j=1}^k
\{1-F_{Z_j}(t;\alpha_j,\lambda)\}^{r_j}
\right]{\rm d}t}
{m\mu}.
\]

Finally,
$
\mu
=
\sum_{j=1}^k\pi_j\mathbb{E}[Z_j]
=
({1}/{\lambda})
\sum_{j=1}^k\pi_j\alpha_j,
$
since $Z_j\sim{\rm Gamma}(\alpha_j,\lambda)$. Substituting this expression for $\mu$ and applying the change of variable $u=\lambda t$ completes the proof.
\end{proof}

		\begin{remark}\label{main-remark}
		For $\boldsymbol{z}=(z_{j,l}:j=1,\ldots,k; l=1,\ldots,r_j)$ and $\boldsymbol{r}=(r_1,\ldots,r_k)^\top$, we define
		\begin{align*}
			f_{\rm max}(\boldsymbol{z},\boldsymbol{r})
			\equiv 
			\max_{j=1,\ldots,k; l=1,\ldots,r_j} z_{j,l}
			\quad 
			\text{and}
			\quad 
			f_{\rm min}(\boldsymbol{z},\boldsymbol{r})
			\equiv 
			\min_{j=1,\ldots,k; l=1,\ldots,r_j} z_{j,l}.
		\end{align*}
			Note that the integrals in Theorem \ref{main-theo-0} admits the following representations:
		\begin{align*}
			\int_0^\infty 
			\left[
			1-
			\prod_{j=1}^{k}	
			{\gamma^{r_j}(\alpha_j,u)\over \Gamma^{r_j}(\alpha_j)}
			\right]
			{\rm d}u
			=
			\mathbb{E}_{\boldsymbol{Z}}\left[f_{\rm max}(\boldsymbol{Z},\boldsymbol{r})\right]
		\end{align*}
		and
		\begin{align*}
			\int_0^\infty 
			\prod_{j=1}^{k}	
			{\Gamma^{r_j}(\alpha_j,u)\over \Gamma^{r_j}(\alpha_j)}
			{\rm d}u
			=
			\mathbb{E}_{\boldsymbol{Z}}\left[f_{\rm min}(\boldsymbol{Z},\boldsymbol{r})\right],
		\end{align*}
		where all the random variables variables $Z_{j,l}\sim \mathrm{Gamma}(\alpha_j,1)$ are independent and $\mathbb{E}_{\boldsymbol{T}}[\cdot]$ denotes the expectation with respect to the random vector $\boldsymbol{T}$.
		
		Furthermore, if we take $\boldsymbol{R}=(R_1,\ldots,R_k)^\top\sim \mathrm{Multinomial}(m;\pi_1,\ldots,\pi_k)$, the formula in Theorem \ref{main-theo-0} can be expressed in the following compact form:
		\begin{align}\label{f-1}
		IG_m
		=
		\left[{1\over m \sum_{j=1}^{k} \pi_j \alpha_j}\right]
		\mathbb{E}_{\boldsymbol{R}}
		\left[
		\mathbb{E}_{\boldsymbol{Z}}\left[f_{\rm max}(\boldsymbol{Z},\boldsymbol{R})\right]
		-
		\mathbb{E}_{\boldsymbol{Z}}\left[f_{\rm min}(\boldsymbol{Z},\boldsymbol{R})\right]
		\right].
		\end{align}
		\end{remark}

	\section{Bias of the $m$th Gini index estimator}\label{Deriving estimator biases}

	In this section (specifically, in Theorem \ref{main-theo}) we derive a general expression 
	for the expectation of the $m$th Gini index estimator, as defined in Section 3 of \cite{VilaSaulomth2025}. 
	For any $m \leqslant n$, the estimator of $IG_m$ is given by
	\begin{equation}\label{estimator}
		\widehat{IG}_m
		=
		\frac{(m-1)!}{(n-1)\cdots(n-m+1)}
		\,
		\frac{
			\displaystyle
			\sum_{1 \leqslant i_1 < \cdots < i_m \leqslant n}
			\left[
			X_{(i_m)}
			-
			X_{(i_1)}
			\right]
		}{
			\displaystyle \sum_{i=1}^{n} X_i
		} \,
        \mathds{1}_{\{\sum_{i=1}^{n} X_i>0\}},
	\end{equation}
	where $X_{(i_1)}=\min\{X_{i_1},\ldots,X_{i_m}\}$ and $X_{(i_m)}=\max\{X_{i_1},\ldots,X_{i_m}\}$ are the extremal order statistics.
	Here, 
	$X_1,\ldots,X_n$ are iid observations from $X \sim \mathrm{GM}(\boldsymbol{\theta})$.
	
	Setting $m=2$ in~\eqref{estimator} yields the usual estimator of the Gini coefficient,
	denoted by $\widehat{G}$:
	\begin{align*}
		\widehat{G}
		\equiv
		\widehat{IG}_2
		=
		\frac{1}{n-1}
		\,
		\frac{
			\displaystyle\sum_{1 \leqslant i < j \leqslant n} |X_i - X_j|
		}{
			\displaystyle\sum_{i=1}^{n} X_i
		}\,
        \mathds{1}_{\{\sum_{i=1}^{n} X_i>0\}},
	\end{align*}
	which was originally proposed by \cite{Deltas2003}.

	Before the main result (Theorem \ref{main-theo}) of this section can be formally stated and proved, we require some preliminary results.

	\begin{proposition}\label{prop-in}
		Let $Y_1$ be a real random variable and $Y_2$ a non-negative random variable with $\mathbb{P}(Y_2 > 0) = 1$. Then
		\begin{align*}
			\mathbb{E}\left[{Y_1\over Y_2} \,
        \mathds{1}_{\{Y_2>0\}}\right]
			=
			\int_0^\infty
			\mathbb{E}\left[Y_1 \exp\left\{-Y_2 z\right\}\right]
			{\rm d}z,
		\end{align*}
		provided that both the expectations and the improper integral are well-defined.
	\end{proposition}
	\begin{proof}
		Applying the identity
		$
		\int_{0}^\infty \exp(-w z){\rm d}z={1/w},
		\ w>0,
		$
		and substituting $w=Y_2$, it follows that
		\begin{align*}
			\mathbb{E}
			\left[
			\dfrac{Y_1}{Y_2}  \,
        \mathds{1}_{\{Y_2>0\}}
			\right]
			=
			\mathbb{E}\left[Y_1 \int_0^\infty\exp\left\{-Y_2 z\right\}{\rm d}z\right],
		\end{align*}
		from which it follows that the integrals may be interchanged, as justified by Tonelli's theorem. The proof is thus complete.
	\end{proof}
	
	\begin{lemma}\label{lemma-1}
		If \(X_1, \dots, X_n\) are independent copies of \(X \sim \mathrm{GM}(\bm{\theta})\), then
		\begin{multline*}
			\mathbb{E}
			\left[
			\dfrac{X_{(1)}}{ \sum_{i=1}^n X_i}   \,
        \mathds{1}_{\{\sum_{i=1}^{n} X_i>0\}}
			\right]
			= 
			\sum_{\substack{r_1 + \cdots + r_k = m \\
					r_1, \ldots, r_k\geqslant 0}}
			\binom{m}{r_1, \ldots, r_k}
			\left(\prod_{j=1}^{k}	
			\pi_j^{r_j}\right)
			\\[0,2cm]
			\times 
			\int_0^\infty 
			\prod_{j=1}^{k}	
			{\Gamma^{r_j}(\alpha_j,u)\over \Gamma^{r_j}(\alpha_j)}
			{\rm d}u
			\left[
			\sum_{\substack{s_1 + \cdots + s_k = n-m  \\ s_1,\ldots,s_k\geqslant 0
			}}
			\binom{n-m}{s_1, \ldots, s_k}
			\left(
			\prod_{j=1}^{k}	
			\pi_j^{s_j}
			\right)
			{1	\over \sum_{j=1}^{k}	\alpha_j (r_j+s_j)}
			\right],
		\end{multline*}
		where $X_{(1)}=\min\{X_1,\ldots,X_m\}$.
	\end{lemma}
	\begin{proof}
Applying Proposition \ref{prop-in} with
$Y_1=X_{(1)}$ and $Y_2=\sum_{i=1}^nX_i$,
\[
\mathbb{E}\!\left[\frac{X_{(1)}}{\sum_{i=1}^nX_i} \,
        \mathds{1}_{\{\sum_{i=1}^{n} X_i>0\}}\right]
=
\int_0^\infty
\mathbb{E}\!\left[
X_{(1)}
\exp\left\{-z\sum_{i=1}^nX_i\right\}
\right]{\rm d}z.
\]

Using the representation
\[
X_{(1)}
=
\int_0^\infty
\mathds 1_{\{X_1\geqslant t,\ldots,X_m\geqslant t\}}\,{\rm d}t,
\]
Tonelli's theorem and the independence of
$X_1,\ldots,X_n$, we obtain
\[
\mathbb{E}\!\left[\frac{X_{(1)}}{\sum_{i=1}^nX_i} \,
        \mathds{1}_{\{\sum_{i=1}^{n} X_i>0\}}\right]
=
\int_0^\infty\!\!\int_0^\infty
\mathbb{E}^m
\!\left[
\mathds 1_{\{X\geqslant t\}}\exp\{-zX\}
\right]
\mathscr L_F^{\,n-m}(z)
\,{\rm d}t\,{\rm d}z.
\]

Since $X\sim\mathrm{GM}(\boldsymbol\theta)$,
Proposition \ref{proposition:ev-gx-11} yields
\[
\mathbb{E}
\!\left[
\mathds 1_{\{X\geqslant t\}}\exp\{-zX\}
\right]
=
\sum_{j=1}^k
\pi_j \, 
\frac{\lambda^{\alpha_j}}
{(z+\lambda)^{\alpha_j}} \, 
\frac{\Gamma(\alpha_j,(z+\lambda)t)}
{\Gamma(\alpha_j)},
\]
and
\[
\mathscr L_F(z)
=
\sum_{j=1}^k
\pi_j \, 
\frac{\lambda^{\alpha_j}}
{(z+\lambda)^{\alpha_j}}.
\]

Substituting these expressions into the previous identity, expanding both powers by the multinomial theorem, and collecting equal powers of $(z+\lambda)$ give the required multiple sum. Finally, the change of variable
$u=(z+\lambda)t$
completes the proof.
\end{proof}

	\begin{lemma}\label{lemma-2}
		If \(X_1, \dots, X_n\) are independent copies of \(X \sim \mathrm{GM}(\bm{\theta})\), then
		\begin{multline*}
			\mathbb{E}
			\left[
			\dfrac{X_{(m)}}{ \sum_{i=1}^n X_i}  \,
        \mathds{1}_{\{\sum_{i=1}^{n} X_i>0\}} 
			\right]
			= 
			\sum_{\substack{r_1 + \cdots + r_k = m \\
					r_1, \ldots, r_k\geqslant 0}}
			\binom{m}{r_1, \ldots, r_k}
			\left(\prod_{j=1}^{k}	
			\pi_j^{r_j}\right)
			\\[0,2cm]
			\times 
			\int_0^\infty 
			\left[
			1-
			\prod_{j=1}^{k}	
			{\gamma^{r_j}(\alpha_j,u)\over \Gamma^{r_j}(\alpha_j)}
			\right]
			{\rm d}u
			\left[
			\sum_{\substack{s_1 + \cdots + s_k = n-m  \\ s_1,\ldots,s_k\geqslant 0
			}}
			\binom{n-m}{s_1, \ldots, s_k}
			\left(
			\prod_{j=1}^{k}	
			\pi_j^{s_j}
			\right)
			{1	\over \sum_{j=1}^{k}	\alpha_j (r_j+s_j)}
			\right],
		\end{multline*}
		where $X_{(m)}=\max\{X_1,\ldots,X_m\}$.
	\end{lemma}
    \begin{proof}
The proof is identical to that of Lemma \ref{lemma-1}. Indeed, replacing the representation
\[
X_{(1)}
=
\int_0^\infty
\mathds 1_{\{X_1\geqslant t,\ldots,X_m\geqslant t\}}\,{\rm d}t
\]
by
\[
X_{(m)}
=
\int_0^\infty
\left[
1-
\mathds 1_{\{X_1\leqslant t,\ldots,X_m\leqslant t\}}
\right]{\rm d} t,
\]
we obtain
\[
\mathbb{E}\!\left[\frac{X_{(m)}}{\sum_{i=1}^nX_i} \,
        \mathds{1}_{\{\sum_{i=1}^{n} X_i>0\}}\right]
=
\int_0^\infty\!\!\int_0^\infty
\left\{
\mathscr L_F^n(z)
-
\mathbb{E}^m
\!\left[
\mathds 1_{\{X\leqslant t\}}\exp\{-zX\}
\right]
\mathscr L_F^{\,n-m}(z)
\right\}
{\rm d}t\,{\rm d}z.
\]

Using
\[
\mathbb{E}
\!\left[
\mathds 1_{\{X\leqslant t\}}\exp\{-zX\}
\right]
=
\sum_{j=1}^k
\pi_j \, 
\frac{\lambda^{\alpha_j}}
{(z+\lambda)^{\alpha_j}} \, 
\frac{\gamma(\alpha_j,(z+\lambda)t)}
{\Gamma(\alpha_j)},
\]
together with the Laplace transform of $X$, the multinomial theorem, and the change of variable $u=(z+\lambda)t$, yields the desired expression.
\end{proof}

	By combining Lemmas \ref{lemma-1} and \ref{lemma-2}, we get 
	\begin{theorem}\label{main-theo}
		Consider independent random variables $X_1,\ldots,X_m$, each distributed as $X \sim \mathrm{GM}(\boldsymbol{\theta})$. Then we have the following:
		\begin{multline*}
	\mathbb{E}[\widehat{IG}_m]
	=
	{1\over m}	
	\sum_{\substack{r_1 + \cdots + r_k = m \\ r_1,\ldots,r_k\geqslant 0}}
	\binom{m}{r_1, \ldots, r_k}
	\left(
	\prod_{j=1}^{k}		
	\pi_j^{r_j}
	\right)
	\int_0^\infty 
	\left[
	1-
	\prod_{j=1}^{k}	
	{\gamma^{r_j}(\alpha_j,u)\over \Gamma^{r_j}(\alpha_j)}
	-
	\prod_{j=1}^{k}	
	{\Gamma^{r_j}(\alpha_j,u)\over \Gamma^{r_j}(\alpha_j)}
	\right]
	{\rm d}u
	\\[0,2cm]
	\times 
	\left[
	n
	\sum_{\substack{s_1 + \cdots + s_k = n-m \\ s_1,\ldots,s_k\geqslant 0}}
	\binom{n-m}{s_1, \ldots, s_k}
	\left(	
	\prod_{j=1}^{k}		
	\pi_j^{s_j}
	\right)
	{
		1\over \sum_{j=1}^{k}	\alpha_j (r_j+s_j)
	}
	\right].
\end{multline*}
	\end{theorem}
	
	\begin{remark}\label{second-remark}
		For  $\boldsymbol{r}=(r_1,\ldots,r_k)^\top$ and $\boldsymbol{s}=(s_1,\ldots,s_k)^\top$, we define
		\begin{align}\label{def-g}
		g_n(\boldsymbol{r},\boldsymbol{s})
		\equiv
		{n\over \sum_{j=1}^{k}	\alpha_j (r_j+s_{j})}.
		\end{align}		

Furthermore, we can simplify the expression in Theorem \ref{main-theo} by introducing two random vectors: 
$\boldsymbol{R} = (R_1,\ldots,R_k)^\top \sim \mathrm{Multinomial}(m;\pi_1,\ldots,\pi_k)$ 
and 
$\boldsymbol{S} = (S_1,\ldots,S_k)^\top \sim \mathrm{Multinomial}(n-m;\pi_1,\ldots,\pi_k)$.
With these definitions, the formula in Theorem~2 can be written in a more compact form:
	\begin{align}\label{f-2}
	\mathbb{E}[\widehat{IG}_m]
	=
	{1\over m}\,
	\mathbb{E}_{\boldsymbol{R}}
	\left[
	\left\{
	\mathbb{E}_{\boldsymbol{Z}}\left[f_{\rm max}(\boldsymbol{Z},\boldsymbol{R})\right]
	-
	\mathbb{E}_{\boldsymbol{Z}}\left[f_{\rm min}(\boldsymbol{Z},\boldsymbol{R})\right]
	\right\}
	\mathbb{E}_{\boldsymbol{S}}\left[g_n(\boldsymbol{R},\boldsymbol{S})\right]
	\right],
\end{align}
where $\mathbb{E}_{\boldsymbol{T}}[\cdot]$ denotes the expectation with respect to  $\boldsymbol{T}$, and $\boldsymbol{Z}$, $f_{\rm max}(\boldsymbol{z},\boldsymbol{r})$ and $f_{\rm min}(\boldsymbol{z},\boldsymbol{r})$ are given in Remark \ref{main-remark}.		
	\end{remark}
	
	\begin{remark}	
		Theorem \ref{main-theo} demonstrates that, due to the scale invariance of $\widehat{IG}_m$, its expectation is independent of the rate $\lambda$.
	\end{remark}

	Combining Theorems \ref{main-theo-0} and \ref{main-theo}, we obtain the following result.
	\begin{corollary}\label{bias} 
		The bias of $\widehat{IG}_m$ relative to $IG_m$, denoted by $\mathrm{Bias}(\widehat{IG}_m, IG_m)$, can be expressed as
		\begin{multline*}
			\mathrm{Bias}(\widehat{IG}_m, IG_m)
			=
			{1\over m}	
			\sum_{\substack{r_1 + \cdots + r_k = m \\ r_1,\ldots,r_k\geqslant 0}}
			\binom{m}{r_1, \ldots, r_k}
			\left(
			\prod_{j=1}^{k}		
			\pi_j^{r_j}
			\right) 
				\int_0^\infty 
			\left[
			1-
			\prod_{j=1}^{k}	
			{\gamma^{r_j}(\alpha_j,u)\over \Gamma^{r_j}(\alpha_j)}
			-
			\prod_{j=1}^{k}	
			{\Gamma^{r_j}(\alpha_j,u)\over \Gamma^{r_j}(\alpha_j)}
			\right]
			{\rm d}u
			\\[0,2cm]
			\times 
			\left[
	n
\sum_{\substack{s_1 + \cdots + s_k = n-m \\ s_1,\ldots,s_k\geqslant 0}}
\binom{n-m}{s_1, \ldots, s_k}
\left(	
\prod_{j=1}^{k}		
\pi_j^{s_j}
\right)
{
	1\over \sum_{j=1}^{k}	\alpha_j (r_j+s_j)
}
			-
			{1\over \sum_{j=1}^{k} \pi_j \alpha_j}
			\right].
		\end{multline*}
	\end{corollary}
	
	\begin{remark}\label{rem-sup}
		Note that the sample-size dependence of  $\mathrm{Bias}(\widehat{IG}_m, IG_m)$ is entirely captured by the quantity 
		\begin{align*}
			n
		\sum_{\substack{s_1 + \cdots + s_k = n-m \\ s_1,\ldots,s_k\geqslant 0}}
		\binom{n-m}{s_1, \ldots, s_k}
		\left(	
		\prod_{j=1}^{k}		
		\pi_j^{s_j}
		\right)
		{
			1\over \sum_{j=1}^{k}	\alpha_j (r_j+s_j)
		}
		=
		\mathbb{E}\left[g_n(\boldsymbol{r},\boldsymbol{S})\right],
		\end{align*}
		where $g_n(\boldsymbol{r},\boldsymbol{s})$ is as given in \eqref{def-g} and $\boldsymbol{S} = (S_1,\ldots,S_k)^\top \sim \mathrm{Multinomial}(n-m;\pi_1,\ldots,\pi_k)$ is as in Remark \ref{second-remark}.
	\end{remark}

	\begin{remark}
		With the exception of $\alpha_j=1$ for $j=1,\ldots,k$, observe that the integrals in Theorem \ref{main-theo} and Corollary \ref{bias} are not representable in closed form using standard mathematical functions and therefore need to be computed numerically.
	\end{remark}
	
	\begin{remark}
		Setting $\alpha_j = \alpha$ for $j = 1, \dots, k$ in Corollary \ref{bias} yields 
		\[
		\mathrm{Bias}(\widehat{IG}_m, IG_m) = 0,
		\] 
		thereby validating the unbiasedness of the $m$th Gini index estimator for gamma populations as established in \cite{VilaSaulomth2025}. In the special case $m=2$, we recover the unbiasedness result previously established by \cite{Baydil2025}.
	\end{remark}

	\begin{remark}
		By combining formulas in \eqref{f-1} and \eqref{f-2} we obtain a compact formula for the bias of $\widehat{IG}_m$ relative to $IG_m$:
		\begin{align*}
\mathrm{Bias}(\widehat{IG}_m, IG_m)
=
	{1\over m}\,
\mathbb{E}_{\boldsymbol{R}}
\left[
\left\{
\mathbb{E}_{\boldsymbol{Z}}\left[f_{\rm max}(\boldsymbol{Z},\boldsymbol{R})\right]
-
\mathbb{E}_{\boldsymbol{Z}}\left[f_{\rm min}(\boldsymbol{Z},\boldsymbol{R})\right]
\right\}
\left\{
\mathbb{E}_{\boldsymbol{S}}\left[g_n(\boldsymbol{R},\boldsymbol{S})\right]
-
{1\over \sum_{j=1}^{k} \pi_j \alpha_j}
\right\}
\right],
		\end{align*}
		where where $\mathbb{E}_{\boldsymbol{T}}[\cdot]$ denotes the expectation with respect to  $\boldsymbol{T}$, 
        $\boldsymbol{Z}$, $f_{\rm max}(\boldsymbol{z},\boldsymbol{r})$ and $f_{\rm min}(\boldsymbol{z},\boldsymbol{r})$ are given in Remark \ref{main-remark}, and $\boldsymbol{R}$, $\boldsymbol{S}$ and $g_n(\boldsymbol{r},\boldsymbol{s})$ are given in Remark \ref{second-remark}.
	\end{remark}

	\section{
    Asymptotic lower bound for the bias} \label{Asymptotic lower bound for the bias}

    This section establishes an asymptotic lower bound for the bias of
$\widehat{IG}_m$. Using the multinomial structure of
$\boldsymbol{S}$ together with Jensen's inequality, we first derive a
lower bound for $\mathbb{E}[g_n(\boldsymbol{r},\boldsymbol{S})]$.
Combining this result with the bias representation obtained in
Corollary~\ref{bias} leads to an asymptotic lower bound for
$\mathrm{Bias}(\widehat{IG}_m,IG_m)$.
	
Indeed, given that
$
\boldsymbol{S}=(S_{1},\ldots,S_{k})^\top
\sim
\mathrm{Multinomial}(n-m;\pi_1,\ldots,\pi_k)
$
and
$
\mathbb{E}[S_j]=(n-m)\pi_j,
$
we have
\begin{align*}
\mathbb{E}\!\left[\frac{1}{g_n(\boldsymbol{r},\boldsymbol{S})}\right]
=
\frac{\sum_{j=1}^{k}\alpha_j(r_j+\mathbb{E}[S_j])}{n}
=
\sum_{j=1}^{k}\pi_j\alpha_j
+
{\rm o}(1),
\end{align*}
where $g_n(\boldsymbol{r},\boldsymbol{s})$ is defined in \eqref{def-g} and
\[
{\rm o}(1)
=
\frac{\sum_{j=1}^{k}\alpha_j(r_j-m\pi_j)}{n}.
\]

Since $x\mapsto1/x$ is convex on $(0,\infty)$, Jensen's inequality yields
\[
\frac{1}{\mathbb{E}[g_n(\boldsymbol{r},\boldsymbol{S})]}
\leqslant 
\mathbb{E}\!\left[\frac{1}{g_n(\boldsymbol{r},\boldsymbol{S})}\right]
=
\sum_{j=1}^{k}\pi_j\alpha_j
+
{\rm o}(1).
\]
Equivalently,
\begin{equation}\label{non-negative-assymp}
\mathbb{E}[g_n(\boldsymbol{r},\boldsymbol{S})]
\geqslant
\frac{1}
{\sum_{j=1}^{k}\pi_j\alpha_j+{\rm o}(1)}.
\end{equation}

Since, by Corollary~\ref{bias} and Remark~\ref{rem-sup},
\[
\mathrm{Bias}(\widehat{IG}_m,IG_m)
\geqslant 
C_m
\left(
\frac{1}
{\sum_{j=1}^{k}\pi_j\alpha_j+{\rm o}(1)}
-
\frac{1}{\sum_{j=1}^{k}\pi_j\alpha_j}
\right),
\]
where
\[
C_m
=
\frac{1}{m}
\sum_{\substack{r_1+\cdots+r_k=m\\ r_1,\ldots,r_k\geqslant 0}}
\binom{m}{r_1,\ldots,r_k}
\left(\prod_{j=1}^{k}\pi_j^{r_j}\right)
\int_{0}^{\infty}
\left[
1
-
\prod_{j=1}^{k}
\frac{\gamma^{r_j}(\alpha_j,u)}{\Gamma^{r_j}(\alpha_j)}
-
\prod_{j=1}^{k}
\frac{\Gamma^{r_j}(\alpha_j,u)}{\Gamma^{r_j}(\alpha_j)}
\right]
\,\mathrm{d}u.
\]

    \begin{remark}
Let
$
\alpha_{(1)}
=
\min\{\alpha_1,\ldots,\alpha_k\}.
$
Since
\begin{align*}
\frac{\sum_{j=1}^{k}\alpha_j(r_j+s_j)}{n}
\frac{\sum_{j=1}^{k}\alpha_jr_j}{n}
+
\frac{\alpha_{(1)}(n-m)}{n}  
=
\alpha_{(1)}
+
\delta_n,
\end{align*}
where
\[
\delta_n
=
\frac{\sum_{j=1}^{k}(\alpha_j-\alpha_{(1)})r_j}{n}
\geqslant 0,
\quad
\delta_n\longrightarrow0,
\]
we obtain
\begin{equation}\label{id-22}
g_n(\boldsymbol{r},\boldsymbol{s})
=
\frac{n}
{\sum_{j=1}^{k}\alpha_j(r_j+s_j)}
\leqslant 
\frac{1}{\alpha_{(1)}+\delta_n}
\leqslant 
\frac{1}{\alpha_{(1)}}.
\end{equation}
Hence, combining \eqref{non-negative-assymp} and \eqref{id-22} yields
\[
\frac{1}
{\sum_{j=1}^{k}\pi_j\alpha_j+{\rm o}(1)}
\leqslant 
\mathbb{E}[g_n(\boldsymbol{r},\boldsymbol{S})]
\leqslant 
\frac{1}{\alpha_{(1)}+\delta_n}
\leqslant 
\frac{1}{\alpha_{(1)}}.
\]
\end{remark}

	\section{Asymptotic unbiasedness}\label{Asymptotic unbiasedness}
	
	In this section, we show that the bias of the $m$th Gini index estimator, as given in Corollary \ref{bias},  vanishes as the sample size $n$ grows (see Theorem \ref{bias-vanishes}). The argument is based on the following technical result.
	
	\begin{proposition} \label{conv-prob}
		It holds that:
		\begin{align*}
		g_n(\boldsymbol{r},\boldsymbol{S})
			\stackrel{\rm P}{\longrightarrow}
			{1\over 
				\sum_{j=1}^{k} \pi_j \alpha_j},
			\quad \text{as} \ n\to\infty,
		\end{align*}
		where $g_n(\boldsymbol{r},\boldsymbol{s})$ is as defined in \eqref{def-g}, $\boldsymbol{S} = (S_1,\ldots,S_k)^\top \sim \mathrm{Multinomial}(n-m;\pi_1,\ldots,\pi_k)$ and ``$\stackrel{{\rm P}}{\longrightarrow}$" denotes convergence in probability.
	\end{proposition}
	\begin{proof}
		Since	$\boldsymbol{S}\sim \textrm{Multinomial}(n-m;\pi_1,\ldots,\pi_k)$, by weak law of large numbers,
		\begin{align*}
			{\boldsymbol{S}\over n-m} \stackrel{\rm P}{\longrightarrow}
			(\pi_1,\ldots,\pi_k)^\top.
		\end{align*}
		Applying the continuous mapping theorem, we get
		\begin{align}\label{conv-in}
			{1\over n-m}
			\sum_{j=1}^{k}	\alpha_jS_j
			\stackrel{\rm P}{\longrightarrow}
			\sum_{j=1}^{k} \pi_j \alpha_j.
		\end{align}
		As
		\begin{align*}
			{1\over n-m}
			\sum_{j=1}^{k}	\alpha_j (r_j+S_{j})
			=
			{1\over n-m}
			\sum_{j=1}^{k}	\alpha_j r_j
			+
			{1\over n-m}
			\sum_{j=1}^{k}	\alpha_j S_{j}
		\end{align*}
		and
		$
			\lim_{n\to\infty}
			\sum_{j=1}^{k}	\alpha_j r_j/(n-m)
			=0,
		$
		by using \eqref{conv-in} and by applying Slutsky's theorem we have
		\begin{align*}
			{1\over n-m}
			\sum_{j=1}^{k}	\alpha_j (r_j+S_{j})
			\stackrel{\rm P}{\longrightarrow}
			\sum_{j=1}^{k} \pi_j \alpha_j.
		\end{align*}
		Again, continuous mapping theorem gives 
		\begin{align}\label{conv-in-1}
			{n-m\over 
				\sum_{j=1}^{k}	\alpha_j (r_j+S_{j})
			}
			\stackrel{\rm P}{\longrightarrow}
			{1\over 
				\sum_{j=1}^{k} \pi_j \alpha_j}.
		\end{align}
		As
		\begin{align*}
			g_n(\boldsymbol{r},\boldsymbol{S})
			=
			{n\over \sum_{j=1}^{k}	\alpha_j (r_j+S_{j})}
			=
			{n\over n-m} \, 
			{n-m\over 
				\sum_{j=1}^{k}	\alpha_j (r_j+S_{j})}
		\end{align*}
		and 
		$
			\lim_{n\to\infty}
			{n/(n-m)}
			=1,
		$
		by combining \eqref{conv-in-1} with Slutsky's theorem yields the desired convergence.
		This concludes the proof.
	\end{proof}
	
	\begin{theorem} \label{bias-vanishes}
		It holds that:
		\begin{align*}
			\lim_{n\to\infty}\mathrm{Bias}(\widehat{IG}_m, IG_m)=0.
		\end{align*}
	\end{theorem}
    \begin{proof}
From Corollary \ref{bias}, Remark \ref{rem-sup}, and the convergence (see Proposition \ref{conv-prob})
\[
g_n(\boldsymbol{r},\boldsymbol{S})
\stackrel{\mathrm{P}}{\longrightarrow}
\frac{1}{\sum_{j=1}^{k}\pi_j\alpha_j},
\quad n\to\infty,
\]
it is enough to prove that the sequence
$
\left\{
g_n(\boldsymbol{r},\boldsymbol{S})
\right\}_{n\ge1}
$
is uniformly integrable. Indeed, convergence in probability implies convergence in distribution, and convergence in distribution together with uniform integrability yields convergence of expectations.

By \eqref{id-22},
\[
0
\leqslant 
g_n(\boldsymbol{r},\boldsymbol{S})
\leqslant 
\frac{1}{\alpha_{(1)}}
\quad \text{a.s. for every } n\geqslant 1.
\]
Hence,
\[
\sup_{n\geqslant 1}
\left\|
g_n(\boldsymbol{r},\boldsymbol{S})
\right\|_{\infty}
\leqslant 
\frac{1}{\alpha_{(1)}}<\infty,
\]
which immediately implies that
$\{g_n(\boldsymbol{r},\boldsymbol{S})\}_{n\geqslant 1}$
is uniformly integrable. Therefore,
\[
\mathbb{E}\!\left[g_n(\boldsymbol{r},\boldsymbol{S})\right]
\longrightarrow
\frac{1}{\sum_{j=1}^{k}\pi_j\alpha_j},
\]
and the desired conclusion follows from Corollary \ref{bias} and Remark \ref{rem-sup}.
\end{proof}

}

\section{Strong consistency and asymptotic normality}\label{Strong consistency and asymptotic normality}

This section establishes the asymptotic properties of
$\widehat{IG}_m$. By representing the estimator as the ratio of a
symmetric $U$-statistic and the sample mean, standard results from the
theory of $U$-statistics are employed to derive its strong consistency
and asymptotic normality.
\begin{theorem}
Assume that $X_1,X_2,\ldots$ are i.i.d. random variables satisfying
$
0<\mu=\mathbb{E}[X]<\infty.
$
Then,
\[
\widehat{IG}_m
\stackrel{\mathrm{a.s.}}{\longrightarrow}
IG_m,
\quad n\to\infty,
\]
where $\stackrel{\mathrm{a.s.}}{\longrightarrow}$ denotes almost sure convergence.
\end{theorem}
\begin{proof}
Write
$
\widehat{IG}_m
=
{U_n}/{\overline X},
$
where
\begin{align}\label{de-U-stat}
U_n
=
\frac{m!}{n(n-1)\cdots(n-m+1)}
\sum_{1\leqslant i_1<\cdots<i_m\leqslant n}
h(X_{i_1},\ldots,X_{i_m}),
\quad 
h(x_{1},\ldots,x_{m})=x_{m:m}-x_{1:m},
\end{align}
is a symmetric $U$-statistic of order $m$, and
$
\overline X=(1/n)\sum_{i=1}^nX_i.
$

By the strong law of large numbers for $U$-statistics \citep{ Lee1990,Henze2024},
$
U_n
\stackrel{\mathrm{a.s.}}{\longrightarrow}
\mathbb{E}\!\left[
X_{m:m}-X_{1:m}
\right].
$
Moreover, by the classical strong law of large numbers,
$
\overline X
\stackrel{\mathrm{a.s.}}{\longrightarrow}
\mu.
$
Since $\mu\in(0,\infty)$, the continuous mapping theorem 
concludes the proof.
\end{proof}

\begin{theorem}
Assume that
$
\mathbb{E}[X^2]<\infty.
$
Then,
\[
\sqrt n
(
\widehat{IG}_m-IG_m
)
\stackrel{\mathscr D}{\longrightarrow}
N(0,\chi^2),
\]
where
$\stackrel{\mathscr D}{\longrightarrow}$ denotes convergence in distribution, 
$
\chi^2
=
\nabla g(\vartheta,\mu)^\top
{\bm \Sigma}
\nabla g(\vartheta,\mu),
$
with
$
g(x,y)={x}/{y},
$
$\boldsymbol{\Sigma}$ denotes the asymptotic covariance matrix of
$
\sqrt n
(
U_n-\vartheta,\,
\overline X-\mu
),
$
and 
$\vartheta
=
\mathbb{E}[X_{m:m}-X_{1:m}]$.
\end{theorem}
\begin{proof}
Note that $\widehat{IG}_m$ can be written as
$
\widehat{IG}_m
=
g(U_n,\overline X),
$
where $U_n$
is the $U$-statistic with kernel defined in \eqref{de-U-stat}.
Since $\mathbb{E}[X^2]<\infty$, 
$h(X_1,\ldots,X_m)$ has finite second moment, and the joint central limit theorem for $U$-statistics of \citet[Theorem 7.3]{Hoeffding1948} and sample means yields
\[
\sqrt n
(
U_n-\vartheta,
\,
\overline{X}-\mu
)
\stackrel{\mathscr{D}}{\longrightarrow}
N_2(\boldsymbol{0},\boldsymbol{\Sigma}),
\]
where
$
\vartheta=
\mathbb{E}[X_{m:m}-X_{1:m}],
$
and $\boldsymbol{\Sigma}$ denotes the corresponding asymptotic covariance matrix.

Since $\mu>0$, the function
$
g(x,y)={x}/{y}
$
is continuously differentiable in a neighborhood of $(\vartheta,\mu)$, with gradient $
\nabla g(\vartheta,\mu)
=
\left(
1/\mu,
-{\vartheta}/{\mu^2}
\right)^\top
$.

Therefore, the multivariate delta method gives
\[
\sqrt n
(
\widehat{IG}_m-IG_m
)
\stackrel{\mathscr D}{\longrightarrow}
N(
0,
\chi^2
),
\]
which completes the proof.
\end{proof}

\section{Illustrative simulation study}\label{Illustrative simulation study}

In this section, Monte Carlo simulations are presented to assess the estimator $\widehat{IG}_{m}$ of the $m$th Gini index $IG_m$.
We analyze replications of the random sample $X_1^{(i)}, \ldots, X_n^{(i)}$ $(i=1,\ldots,M)$ generated by a $\operatorname{MG}(\boldsymbol{\theta})$ distribution (Definition \ref{definition:mixture-gamas}), $\boldsymbol{\theta}=(\pi_1, \cdots, \pi_{k-1}, \alpha_1,\cdots, \alpha_k, \lambda)^\top$, $\alpha_j>0$, $\lambda=1/\beta>0$, $\pi_j>0$, $\pi_k=1-\sum_{j=1}^{k-1} \pi_j$.
We assess the performance of the estimator using the Bias and the root mean square error (RMSE) 
$$ \operatorname{Bias} = \frac{1}{M} \sum_{i=1}^M \left[\widehat{IG}_{m}^{(i)} - IG_m(\boldsymbol{\theta}) \right]~~\text{and}~~ RMSE = \sqrt{\frac{1}{M} \sum_{i=1}^M \left[\widehat{IG}_{m}^{(i)} - IG_m(\boldsymbol{\theta}) \right]^2},$$ 
under different parameter scenarios. Here
 $\widehat{IG}_{m}^{(i)}=\widehat{IG}_{m}(X_1^{(i)},\ldots,X_n^{(i)})$ is given in \eqref{estimator}.
 
We also examine the sensitivity of the estimator to different values of $m$, and compare the results with the competitors' estimators: the bias-corrected (based on  Corollary~\ref{bias})
$$\widehat{IG}_m^{BC} = \widehat{IG}_m - \operatorname{Bias}(\widehat{IG}_m, IG_m),$$
and the parametric estimator $ IG_m(\widehat{\boldsymbol{\theta}})$ (based on Theorem~\ref{main-theo-0}), where $\widehat{\boldsymbol{\theta}}$ is the maximum likelihood estimate of the $\boldsymbol{\theta}$ parameter.

\subsection{Evaluating the Bias and RMSE of estimates}

We conducted $M=1,000$ Monte Carlo replications, generating samples $X_1, X_2,\ldots, X_n$ from the $\operatorname{GM}(\boldsymbol{\theta}$) considering for the simulations the true population's parameter $\pi_1=0.6$, $\alpha_1=0.5$, $\alpha_2\in\{0.5, 1,2,5\}$, $\beta=1.5$,  and the simple sizes $n\in\{10, 20, 50\}$. A sensitivity analysis of the results with respect to the choice of $m\in\{2,4\}$ is also presented.

Table \ref{tab:mc_alpha} presents the Bias and RMSE of the following estimators: $m$th Gini $\widehat{IG}_m$, the bias-corrected $\widehat{IG}_m^{BC}$, and the parametric estimator $IG_m(\widehat{\boldsymbol{\theta}})$. 
Here,
\begin{equation}\label{eq:IGm_mle}
    IG_m(\widehat{\boldsymbol{\theta}}) = 
			\dfrac{\displaystyle
				\int_0^\infty 
				\left\{
				1-F^m_X(t; \widehat{\bm{\theta}})
				-
				[
				1-F_X(t; \widehat{\bm{\theta}})
				]^m
				\right\}
				{\rm d}t
			}
			{\displaystyle
				m\mu 
			}
\end{equation}
is computed numerically. 
Observe that as the sample size $n$ increases, both the Bias and the RMSE of the Monte Carlo simulations generally decrease. 

The bias of $\widehat{IG}_m$ is small. However, for small sample sizes, $\widehat{IG}_m^{BC}$ becomes more accurate, exhibiting lower bias, as expected. For moderate sample sizes ($n=50$), $\widehat{IG}_m$ performs as well as $\widehat{IG}_m^{BC}$, which is also expected since bias (Corollary \ref{bias}) decays to zero as $n$ increases.
The parametric estimator, although yielding good results, was the one that presented the largest biases. However, across all scenarios analyzed, $IG_m(\widehat{\boldsymbol{\theta}})$ achieved the lowest RMSE. 

\begin{table}[H]
  \centering
 \caption{Bias and RMSE for the estimates of the $\widehat{IG}_m$, $\widehat{IG}_m^{BC}$ and $IG_m(\widehat{\boldsymbol{\theta}})$ based on Monte Carlo simulations from $X \sim \mathrm{GM}(\boldsymbol{\theta})$, with $m \in \{2,4\}$, $n\in\{10,20,50\}$, $\alpha_1=0.5$, $\alpha_2\in\{0.5, 1,2,5\}$, $\pi_1=0.6$, $\beta=1.5$ and $M = 1,000$ replications.}
   \label{tab:mc_alpha}
\begin{tabular}{llllcccccc}
\hline
&&&& \multicolumn{2}{c}{$\widehat{IG}_m$} &
\multicolumn{2}{c}{$\widehat{IG}_m^{BC}$} &
\multicolumn{2}{c}{$IG_m(\widehat{\boldsymbol{\theta}})$} \\
\cline{5-6}\cline{7-8}\cline{9-10}
$m$&$\alpha_2$ & $n$&$IG_m$ &
Bias & RMSE &
Bias & RMSE &
Bias & RMSE \\
\hline
2 & 0.5    & 10 & 0.6366      & -0.0031 & 0.0940 & -0.0120 & 0.0919 & -0.0601 & 0.1051 \\
2 & 0.5    & 20 & 0.6366      & -0.0073 & 0.0710 & -0.0127 & 0.0710 & -0.0373 & 0.0759 \\
2 & 0.5    & 50 & 0.6366      & -0.0033 & 0.0452 & -0.0057 & 0.0457 & -0.0176 & 0.0461 \\
2 & 1      & 10 & 0.5914      & 0.0061  & 0.0953 & -0.0038 & 0.0928 & -0.0487 & 0.1003 \\
2 & 1      & 20 & 0.5914      & 0.0066  & 0.0706 & 0.0007  & 0.0696 & -0.0211 & 0.0704 \\
2 & 1      & 50 & 0.5914      & 0.0018  & 0.0440 & -0.0004 & 0.0439 & -0.0095 & 0.0441 \\
2 & 2      & 10 & 0.5802      & 0.0142  & 0.0988 & 0.0004  & 0.0948 & -0.0390 & 0.0979 \\
2 & 2      & 20 & 0.5802      & 0.0044  & 0.0689 & -0.0033 & 0.0688 & -0.0217 & 0.0691 \\
2 & 2      & 50 & 0.5802      & 0.0004  & 0.0443 & -0.0032 & 0.0442 & -0.0100 & 0.0451 \\
2 & 5      & 10 & 0.6068      & 0.0274  & 0.1008 & 0.0066  & 0.0986 & -0.0322 & 0.0951 \\
2 & 5      & 20 & 0.6068      & 0.0158  & 0.0765 & 0.0027  & 0.0760 & -0.0137 & 0.0708 \\ 
2 & 5      & 50 & 0.6068      & 0.0032  & 0.0438 & -0.0023 & 0.0442 & -0.0080 & 0.0429 \\\hline
4 & 0.5    & 10 & 0.5874      & 0.0039  & 0.0924 & -0.0069 & 0.0891 & -0.0525 & 0.0970 \\
4 & 0.5    & 20 & 0.5874      & -0.0045 & 0.0692 & -0.0120 & 0.0694 & -0.0355 & 0.0722 \\
4 & 0.5    & 50 & 0.5874      & -0.0027 & 0.0451 & -0.0053 & 0.0453 & -0.0174 & 0.0453 \\
4 & 1      & 10 & 0.5435      & 0.0061  & 0.0899 & -0.0049 & 0.0865 & -0.0473 & 0.0921 \\
4 & 1      & 20 & 0.5435      & 0.0050  & 0.0661 & -0.0019 & 0.0655 & -0.0222 & 0.0651 \\
4 & 1      & 50 & 0.5435      & -0.0039 & 0.0431 & -0.0067 & 0.0434 & -0.0158 & 0.0447 \\
4 & 2      & 10 & 0.5306      & 0.0138  & 0.0931 & -0.0001 & 0.0903 & -0.0381 & 0.0894 \\
4 & 2      & 20 & 0.5306      & 0.0050  & 0.0688 & -0.0047 & 0.0678 & -0.0209 & 0.0677 \\
4 & 2      & 50 & 0.5306      & 0.0002  & 0.0427 & -0.0043 & 0.0425 & -0.0098 & 0.0426 \\
4 & 5      & 10 & 0.5495      & 0.0242  & 0.1024 & 0.0017  & 0.1010 & -0.0296 & 0.0916 \\
4 & 5      & 20 & 0.5495      & 0.0142  & 0.0719 & 0.0004  & 0.0709 & -0.0124 & 0.0668 \\
4 & 5      & 50 & 0.5495      & 0.0041  & 0.0442 & -0.0017 & 0.0446 & -0.0058 & 0.0423\\
\hline
\end{tabular}
\end{table}

\subsection{Analysis of the estimator in Gamma mixture with different scales}

In Section \ref{sec:02}, we restrict our theoretical analysis to Gamma mixtures with equal scales \eqref{eq:mixture-gama}, which allows us to derive a closed-form expression for $IG_m$ (Theorem~\ref{main-theo-0}) and its bias (Corollary~\ref{bias}). Moreover, the $m$th Gini index is invariant with respect to the scale parameter. However, real-data applications may involve datasets that require different scales for proper modeling. In this context, it is important to investigate whether the estimators $\widehat{IG}_m$ and $IG_m(\widehat{\boldsymbol{\theta}})$ continue to exhibit low and decreasing bias and RMSE.

Table \ref{tab:mc2} presents the results of Bias and RMSE for the estimates of $\widehat{IG}_m$ and $IG_m(\widehat{\boldsymbol{\theta}})$ based on $M = 1,000$ Monte Carlo simulations from 
$$f_X(x; \bm{\theta}) =  \pi_1 f_{Z_1}(x; \alpha_1, \beta_1)+ (1-\pi_1) f_{Z_2}(x; \alpha_2, \beta_2),~~x > 0,$$
were $Z_j\sim\operatorname{Gamma}(\alpha_j, \beta_j)$.
The parameters chosen are close to real values presented in the next section.
Both estimators performed well, with $\widehat{IG}_m$ being slightly more precise, although both exhibited negative bias values. The bias and RMSE decreased as $n$ increased. The RMSE of the two estimators was very similar across the scenarios considered. Increasing $m$ from 2 to 4 does not change the general results.

\begin{table}[H]
  \centering
 \caption{Bias and RMSE for the estimates of the $\widehat{IG}_m$ and $IG_m(\widehat{\boldsymbol{\theta}})$ based on Monte Carlo simulations from $X \sim \mathrm{GM}(\pi_1, \alpha_1,\alpha_2, \beta_1, \beta_2)$, with $m \in \{2,4\}$, $n\in\{10,20,50\}$, $\alpha_1=1$, $\alpha_2\in\{0.5, 1,2,3\}$, $\pi_1=0.1$, $\beta_1=40$, $\beta_2=8$, and $M = 1,000$ replications.}
   \label{tab:mc2}
\begin{tabular}{llllcccc}
\hline
&&&& \multicolumn{2}{c}{$\widehat{IG}_m$} &
\multicolumn{2}{c}{$IG_m(\widehat{\boldsymbol{\theta}})$} \\
\cline{5-6}\cline{7-8}
$m$&$\alpha_2$ & $n$&$IG_m$ &
Bias & RMSE &
Bias & RMSE \\
\hline
2 & 0.5    & 10 & 0.7362      &$-$0.0387 & 0.1287 &$-$0.1038 & 0.1495 \\
2 & 0.5    & 20 & 0.7362      &$-$0.0270 & 0.0924 &$-$0.0674 & 0.1059 \\
2 & 0.5    & 50 & 0.7362      &$-$0.0091 & 0.0601 &$-$0.0339 & 0.0668 \\
2 & 1      & 10 & 0.5857      &$-$0.0228 & 0.1259 &$-$0.0797 & 0.1336 \\
2 & 1      & 20 & 0.5857      &$-$0.0200 & 0.0941 &$-$0.0543 & 0.0995 \\
2 & 1      & 50 & 0.5857      &$-$0.0099 & 0.0615 &$-$0.0278 & 0.0642 \\
2 & 2      & 10 & 0.4293      &$-$0.0144 & 0.1054 &$-$0.0550 & 0.1068 \\
2 & 2      & 20 & 0.4293      &$-$0.0096 & 0.0807 &$-$0.0328 & 0.0798 \\
2 & 2      & 50 & 0.4293      &$-$0.0044 & 0.0522 &$-$0.0156 & 0.0512 \\
2 & 3      & 10 & 0.3516      &$-$0.0091 & 0.0883 &$-$0.0421 & 0.0880 \\
2 & 3      & 20 & 0.3516      &$-$0.0051 & 0.0617 &$-$0.0223 & 0.0612 \\
2 & 3      & 50 & 0.3516      &$-$0.0006 & 0.0423 &$-$0.0079 & 0.0416 \\
4 & 0.5    & 10 & 0.6948      &$-$0.0399 & 0.1327 &$-$0.1113 & 0.1523 \\
4 & 0.5    & 20 & 0.6948      &$-$0.0327 & 0.1030 &$-$0.0749 & 0.1136 \\
4 & 0.5    & 50 & 0.6948      &$-$0.0151 & 0.0646 &$-$0.0426 & 0.0735 \\
4 & 1      & 10 & 0.5468      &$-$0.0342 & 0.1290 &$-$0.0914 & 0.1368 \\
4 & 1      & 20 & 0.5468      &$-$0.0180 & 0.0991 &$-$0.0542 & 0.1017 \\
4 & 1      & 50 & 0.5468      &$-$0.0092 & 0.0652 &$-$0.0294 & 0.0687 \\
4 & 2      & 10 & 0.3971      &$-$0.0120 & 0.0986 &$-$0.0535 & 0.0970 \\
4 & 2      & 20 & 0.3971      &$-$0.0094 & 0.0767 &$-$0.0322 & 0.0760 \\
4 & 2      & 50 & 0.3971      &$-$0.0044 & 0.0503 &$-$0.0160 & 0.0498 \\
4 & 3      & 10 & 0.3237      &$-$0.0065 & 0.0864 &$-$0.0390 & 0.0839 \\
4 & 3      & 20 & 0.3237      &$-$0.0061 & 0.0618 &$-$0.0234 & 0.0605 \\
4 & 3      & 50 & 0.3237      &$-$0.0013 & 0.0400 &$-$0.0088 & 0.0394\\
\hline
\end{tabular}
\end{table}

\subsection{Computational issues}
We conclude this section by highlighting computational challenges and the strategies adopted in implementing and validating the estimators:

 \begin{itemize}
\item  R implementation and required packages: The simulation study was programmed in R software 4.5.1. 
The incomplete gamma function was obtained using the package \textsf{gsl}.
We also used the native R functions \textsf{integrate}, \textsf{combn}, \textsf{pgamma}, and \texttt{rgamma}.,
For the simulations, we used the \texttt{evmix} package (for random sample generation from Gamma mixtures), mixtools (for parameter estimation), and the packages \texttt{parallel}, \texttt{foreach}, and \texttt{doRNG} (for parallel computations with 7 cores). 

\item For simplicity in using \texttt{mixtools}, the programming equations were expressed in terms of scale parameters $(\beta)$ rather than rate parameters $(\lambda)$, and all formulas were adapted as needed. 

\item As $m$ increases, the growth in combinatorial triples can become computationally demanding, making replications of the estimator extremely costly.
The computational cost associated with $\binom{n}{m}$ also increases as $n$ grows. Our simulation analysis, therefore, focuses on low and medium sample sizes. However, the real data application presented in Section \ref{sec:applications} involves a dataset with a large sample size.

\item Instability in the estimates was observed during the simulation study. In Table~\ref{tab:mc_alpha}, noisy samples were discarded, and new samples were generated to ensure ($M = 1,000$) replications. The percentage of discarded samples decreased as $n$ or $m$ increased, ranging between 4–16\% for ($n = 50$), depending on the parameter choices. Low values of $\alpha_2$ also resulted in higher discard rates. In Table~\ref{tab:mc2}, no samples were discarded. 

\item In Definition \ref{definition:mixture-gamas}, we assume $\lambda_1=\cdots=\lambda_k=\lambda$. That enables us to make explicit mathematical calculations. However, when working with real data, this assumption is quite restrictive and requires a suitable estimation strategy. 
To estimate $\boldsymbol{\theta}$, one can implement an adapted Expectation–Maximization (EM) algorithm (cf. \cite{mclachlan2008algorithm}). 
Another option is to work with different scales, but an estimate of the bias will not be known.
We relied on the information criterion BIC together with graphical analysis to specify the number of components in real data.
 \end{itemize}

\section{Application}\label{sec:applications}
In order to evaluate the $m$th Gini index for gamma mixtures, a public data set of income data was used.

\subsection{Data}
 
 The data\footnote{Available at \url{https://www.bancaditalia.it/statistiche/tematiche/indagini-famiglie-imprese/bilanci-famiglie/documentazione/ricerca/ricerca.html?min_anno_pubblicazione=2008&max_anno_pubblicazione=2008} (accessed on June 2026)} comes from the Bank of Italy's Survey on Household Income and Wealth from the year 2008.
The income is composed of payroll income (net wages, salaries and fringe benefits), pensions and net transfers (pensions, arrears and other transfers), net self-employment income (self-employment income and entrepreneurial income), and property income (either from real estate or from financial assets).

We use the raw data from the survey, available as the data sets RFAM08 (income as $Y$ variable).
The filtered sample yielded 7,958 observations after removing entries whose income is negative.
Table \ref{tab:summary} presents descriptive statistics for the income data. The distance between the maximum and the median, together with the kurtosis (CK) and asymmetry (CS) coefficients, indicates the presence of a heavy right tail.

\begin{table}[H]
\centering
\caption{Descriptive statistics for the income dataset.}
\label{tab:summary}
\begin{tabular}{ccccccccc}
\hline
Min.       & 1st Qu.    & Median     & Mean       & 3rd Qu.    & Max.       & SD         & CS         & CK        \\
\hline
0.0657 & 17.8912 & 26.7837 & 32.4236 & 40.5865 & 629.3397 & 24.3318 & 5.2351 & 75.1661\\
\hline
\end{tabular}
\end{table}

\subsection{Income inequality estimation with Gamma mixtures}
Due to its positive support and suitability for income decomposition, the Gamma mixture is a natural probability model for such data. Table \ref{tab:mth_gini} presents the $IG_m$ estimates using \eqref{estimator} and \eqref{eq:IGm_mle} estimators with $m\in\{2,4\}$ and $k\in\{2,3\}$. Two scenarios are considered for the parametric estimator: equal scales ($\widehat{\boldsymbol{\theta}}$) and different scales ($\widetilde{\boldsymbol{\theta}}$). Here, $\widetilde{\boldsymbol{\theta}}$ is obtained in a two-stage pseudolikelihood estimation.

\begin{table}[H]
\centering
\caption{Estimates of $m$th-Gini index for different values of $m$ and number of componets $k$.} 
\label{tab:mth_gini}

\begin{tabular}{ccccc}
  \hline
  &&&&\\
  $m$ & $k$ & $IG_m(\widehat{\boldsymbol{\theta}})$ & $IG_m(\widetilde{\boldsymbol{\theta}})$ & $\widehat{IG}_m$ \\ 
  \hline
 2 & 2 & 0.3374 & 0.3194 & 0.3399 \\ 
 2 & 3 & 0.3375 & 0.2981 & 0.3399 \\ 
 4 & 2 & 0.3112 & 0.2917 & $(*)$ \\ 
 4 & 3 & 0.3113 & 0.2723 & $(*)$ \\ 
   \hline
\end{tabular}

\medskip
\footnotesize
$(*)$ Not computed because of the large sample size.
\end{table}

\subsection{Model selection}

Table \ref{tab:model_selection} presents the parameter estimates for the Gamma distribution, the Gamma mixtures with two (GM2) and three (GM3) components, and the Gamma mixtures with equal scale parameters (GM2\_ss and GM3\_ss). Minimizing the BIC information criterion suggests that a Gamma mixture with $k=2$ components and different scales would be appropriate. However, Figure \ref{fig:ic_qqplots} indicates that  $k=3$ components are needed to describe the right tail of the randomized residuals and provide a natural interpretation: the first component accounts for low-income values (approximately 88\% of the data). In contrast, the second component (11\%) captures medium- to high-income values, and the last component captures the top 1\% values.

Figure \ref{fig:mixture} shows the PDFs of each component in the mixture. Figure \ref{fig:ic_hist_ecdf} shows the PDF fitted to the histogram and the CDF fitted to the ECDF, both confirming the QQ-plot analysis that the GM3 model adequately describes the income data.

\begin{table}[H]
\centering
\caption{Estimated parameters and information criteria for model selection.} 
\label{tab:model_selection}
\begin{tabular}{lcccc}
  \hline
 Model & Estimates &  Loglik & AIC & BIC \\ 
  \hline
 Gamma &   $(\hat{\alpha}, \hat{\lambda})=(2.65, 0.08)$ & $-$34003.88 & 68011.76 & 68025.72 \\ 
 GM2 & $\widehat{\boldsymbol{\theta}}=(0.89,  3.50, 1.48, 8.37, 39.13)$ & $-$33731.28 & 67472.55 & 67507.46 \\ 
 GM3 & $\widehat{\boldsymbol{\theta}}=(0.88, 0.11, 3.52, 1.51, 458.14, 8.31, 37.32,   0.49)$ & {$-$33730.77} & {67477.54} & {67533.40} \\ 
 GM2\_ss & $\widetilde{\boldsymbol{\theta}}=(0.89, 3.50, 1.48, 9.8)$ & $-$34192.01 & 68392.02 & 68419.95 \\ 
 GM3\_ss & $\widetilde{\boldsymbol{\theta}}=(0.73, 0.055, 4.28, 1.05, 6.57, 7.01)$ & $-$34463.04 & 68938.09 & 68979.98 \\ 
   \hline
\end{tabular}
\end{table}

\begin{figure}[H]
	    \centering
	    \includegraphics[width=1.0\linewidth]{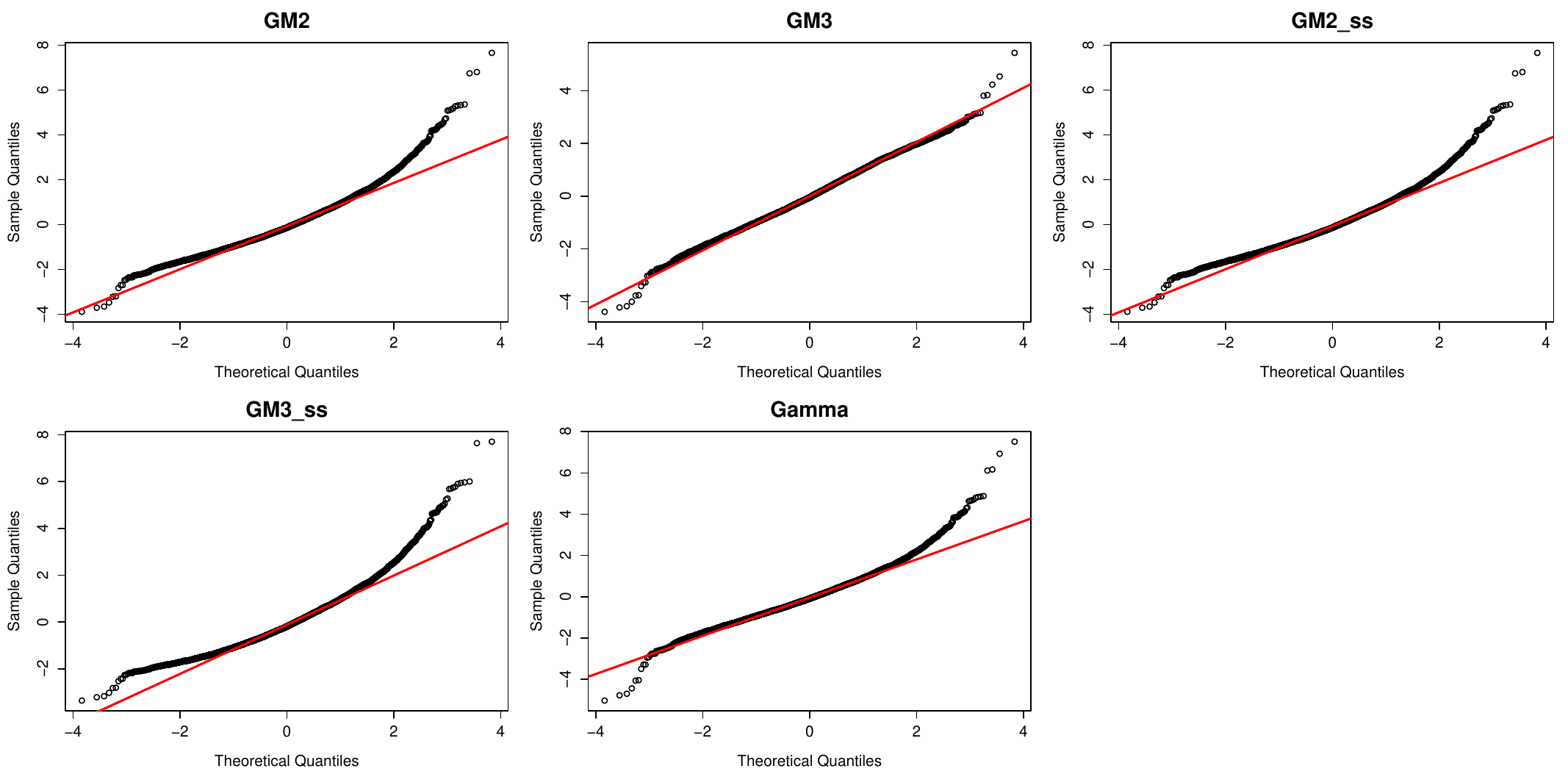}
	    \caption{Quantile-Quantile plot displaying residuals from fitted models for income data.}
	    \label{fig:ic_qqplots}
	\end{figure}

\begin{figure}[H]
	\centering
	\includegraphics[width=1.0\linewidth]{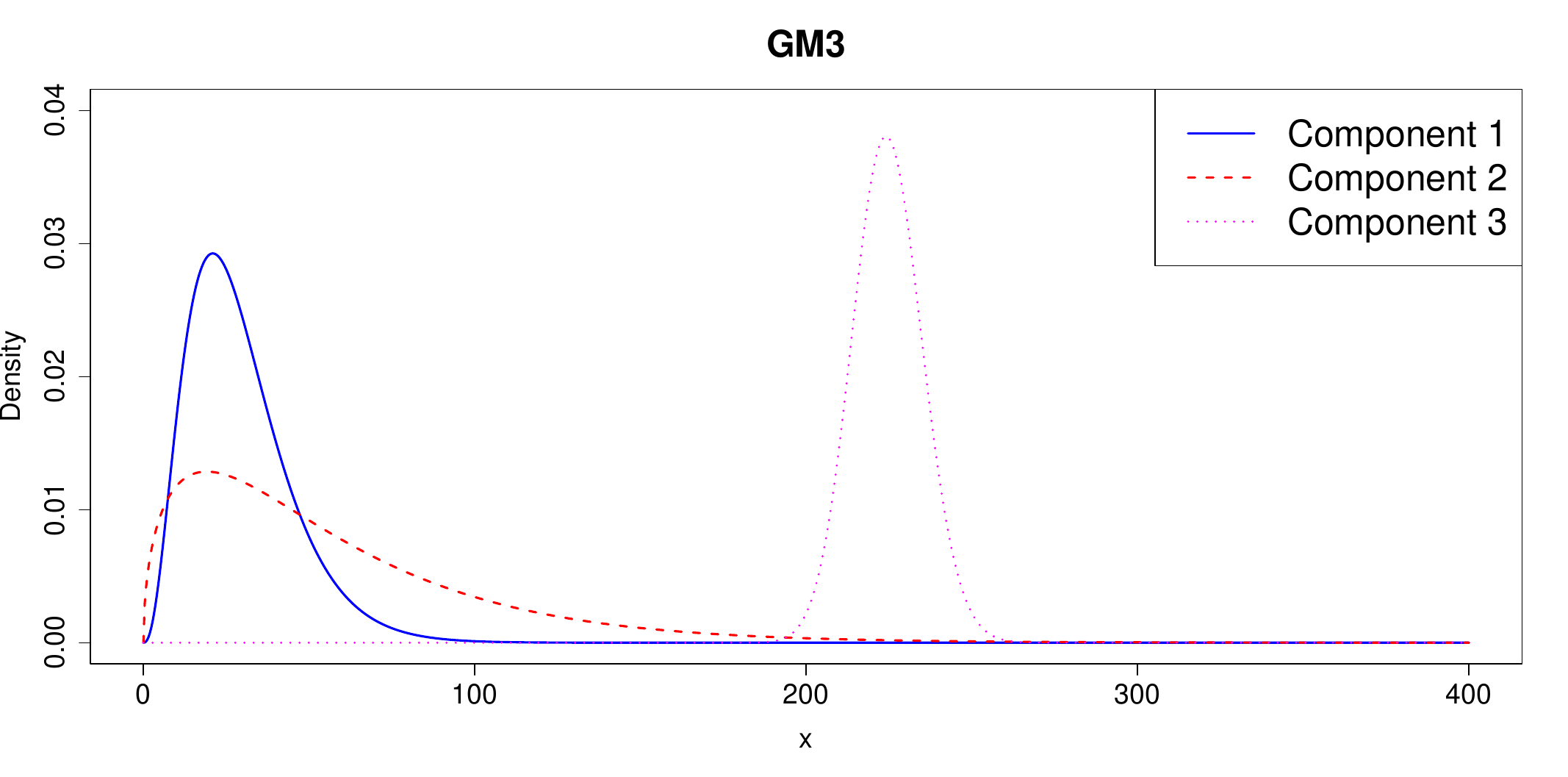}
	\caption{Fitted PDFs (left) and empirical CDFs (right) for income data.}
	\label{fig:mixture}
\end{figure}

\begin{figure}[H]
	\centering
	\includegraphics[width=1.0\linewidth]{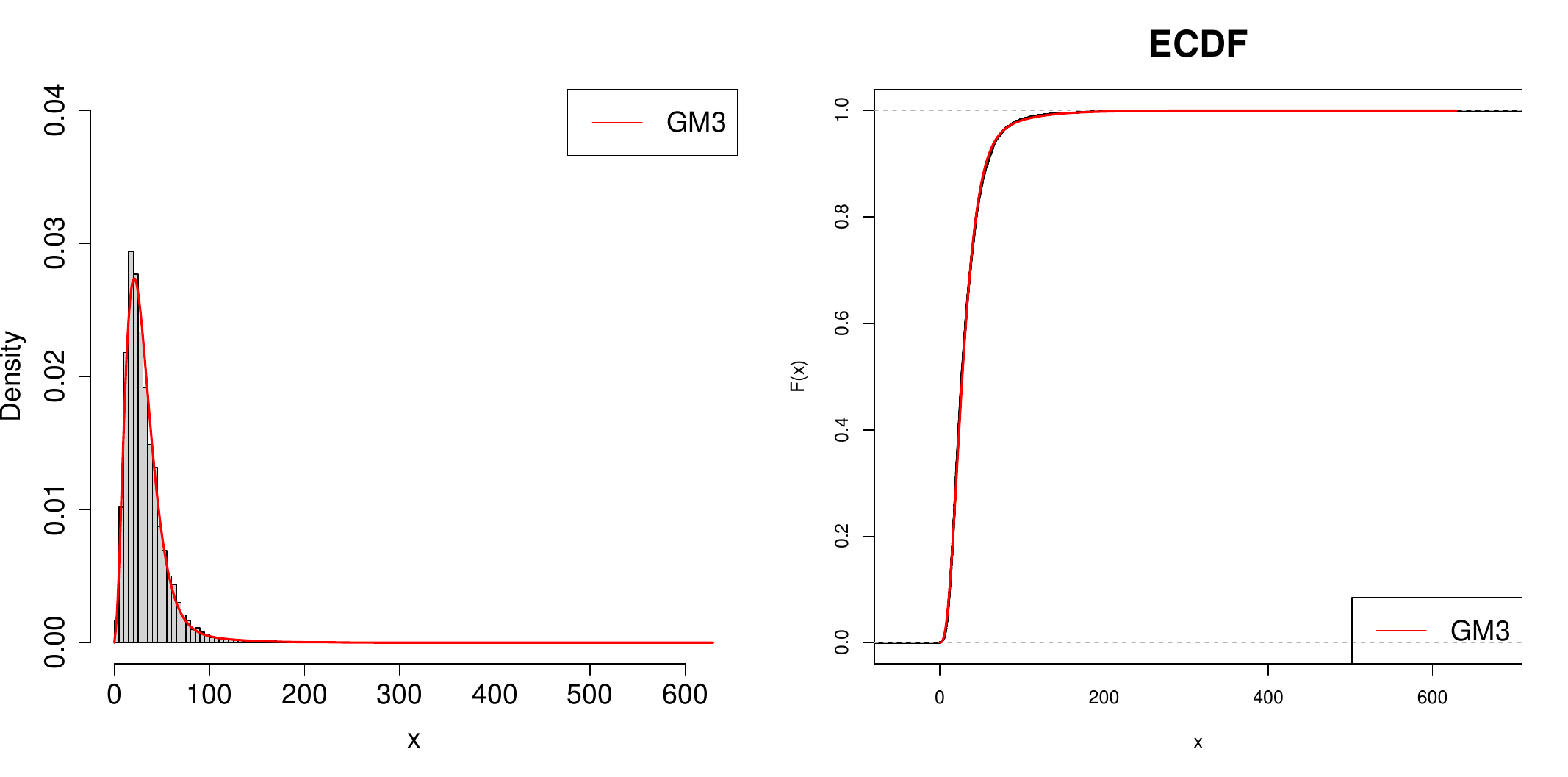}
	\caption{Fitted PDFs (left) and empirical CDFs (right) for income data.}
	\label{fig:ic_hist_ecdf}
\end{figure}


\section{Concluding remarks}\label{concluding_remarks}


This paper investigated the finite-sample behavior of the $m$th Gini index estimator under finite mixtures of gamma distributions with a common rate parameter. We first derived an analytical expression for the population $m$th Gini index and then obtained explicit formulas for the expectation and bias of its nonparametric estimator. These results extend previous unbiasedness and bias expressions available for single gamma populations to a broader and more flexible class of mixture models.

We also established the asymptotic unbiasedness, consistency, and asymptotic normality of the estimator, together with an asymptotic lower bound for its bias. The Monte Carlo study complemented the theoretical developments by comparing the proposed estimator with bias-corrected and parametric alternatives over a range of parameter configurations, including scenarios beyond the assumptions required by the theoretical results. Overall, the simulations indicate that the finite-sample bias decreases with the sample size and illustrate the advantages and limitations of each estimation approach. 

Finally, the analysis of the Bank of Italy income data illustrated the practical applicability of the proposed methodology. By fitting finite gamma mixtures to the observed data, the proposed estimator provided a flexible framework for measuring income inequality while accounting for population heterogeneity. Although the analytical results were derived under a common rate parameter, the empirical analysis also suggests that the proposed methodology remains useful in more general settings. Future research may extend these results to mixtures with distinct rate parameters and to other flexible families of positive distributions.

\section*{Acknowledgments}

The research was supported in part by CNPq and CAPES grants from the Brazilian government.

\section*{Declaration}

There are no conflicts of interest to disclose.



\bibliographystyle{apalike}
\bibliography{paper-ref}

\end{document}